\documentclass[10pt,twocolumn]{IEEEtran}

\usepackage{ifpdf}
\usepackage{manfnt}
\usepackage[mathscr]{eucal}
\usepackage{graphicx}

\usepackage{amssymb}
\usepackage{theorem}
\usepackage{subfigure}
\usepackage{cite}
\usepackage{color}
\usepackage{colortbl}
\definecolor{colorref}{rgb}{0.4648,0,0} 
\definecolor{colorcite}{rgb}{0,0.2902,0.1765}

\usepackage{bm}

\usepackage{amsmath}

\newtheorem{proposition}{Proposition}
\newtheorem{subproposition}{Proposition}[proposition]

\newtheorem{lemma}{Lemma}
\newtheorem{sublemma}{Lemma}[lemma]
\newtheorem{Cor}{Corollary}

\newcommand{\setn}{\mathcal{N}}
\newcommand{\setm}{\mathcal{M}}
\newcommand{\tl}{T_{\mathrm{L}}}
\newcommand{\ta}{T_{\mathrm{A}}}

\newcommand{\Rmnum}[1]{\uppercase\expandafter{\romannumeral #1}}

\renewcommand{\baselinestretch}{1.33}

\newcommand{\signal}[1]{{\boldsymbol{#1}}}
\newcommand{\dprime}{{\prime\prime}}

\newcommand{\real}{{\mathbb R}}
\newcommand{\innerprod}[2]{\left\langle{#1},{#2}\right\rangle}
\newtheorem{definition}{Definition}
\newtheorem{remark}{Remark}
\newtheorem{fact}{Fact}
\newtheorem{subfact}{Fact}[fact]

\newtheorem{example}{Example}
\newcommand{\Natural}{{\mathbb N}}
\newcommand{\refeq}[1]{(\ref{#1})}

\newif\iftodo   
\todotrue
\newif\iftodoshort  
\todoshortfalse

\begin{document}
\title{Elementary Properties of Positive \\Concave Mappings with Applications \\to Network Planning and Optimization}

\author{\IEEEauthorblockN{Renato L. G. Cavalcante, \IEEEmembership{Member, IEEE}, Yuxiang Shen, and S\l{}awomir Sta\'nczak, \IEEEmembership{Senior Member, IEEE} \thanks{This work was partially supported by the Deutsche Forschungsgemeinschaft (DFG) under Grant STA 864\slash9-1. Copyright (c) 2015 IEEE. Personal use of this material is permitted. However, permission to use this material for any other purposes must be obtained from the IEEE by sending a request to pubs-permissions@ieee.org}\\
\IEEEauthorblockA{Fraunhofer Heinrich Hertz Institute / Technical University of Berlin \\
Berlin, Germany \\ 
Email: renato.cavalcante@hhi.fraunhofer.de, yuxiangshenee@gmail.com, slawomir.stanczak@hhi.fraunhofer.de}}}

\maketitle

\begin{abstract}
This study presents novel methods for computing fixed points of positive concave mappings and for characterizing the existence of fixed points. These methods are particularly important in planning and optimization tasks in wireless networks. For example, previous studies have shown that the feasibility of a network design can be quickly evaluated by computing the fixed point of a concave mapping that is constructed based on many environmental and network control parameters such as the position of base stations, channel conditions, and antenna tilts. To address this and more general problems, given a positive concave mapping, we show two alternative but equivalent ways to construct a matrix that is guaranteed to have spectral radius strictly smaller than one if the mapping has a fixed point. This matrix is then used to build a new mapping that preserves the fixed point of the original positive concave mapping. We show that the standard fixed point iterations using the new mapping converges faster than the standard iterations applied to the original concave mapping. As exemplary applications of the proposed methods, we consider the problems of power and load estimation in networks based on the orthogonal frequency division multiple access (OFDMA) technology. 
 \end{abstract}

\section{Introduction}

Problems that can be posed as that of finding fixed points of standard interference mappings are ubiquitous in communication systems \cite{yates95,slawomir09,martin11,Majewski2010,feh2013,siomina12,renato14SPM,ho2015,renato14SIP}, and, in particular, in planning and optimization of networks based on the orthogonal frequency division multiple access (OFDMA) technology \cite{MajewskiTurkeHuangEtAl2007Analytical,Majewski2010,feh2013,siomina12,renato14SPM,ho2015,renato14SIP}. In many of these applications, the mappings are positive concave mappings, which are a strict subclass of standard interference mappings \cite{renato14SPM}.

 For example, by using standard interference coupling models that are widely used in the literature \cite{mogensen07}, the studies in \cite{Majewski2010} and \cite{siomina12} consider a very particular case of a positive concave mapping for the problem of load estimation in long-term evolution (LTE) networks. The fixed point of that mapping, if it exists, indicates the bandwidth required by each base station to satisfy the data rate requirements of users. With this knowledge, we can evaluate the feasibility of a network design by verifying whether the required bandwidth does not exceed the available bandwidth. However, especially in large-scale planning, computation of the fixed point may require time-consuming iterative methods. Therefore, the development of fast tools to ensure the existence of a fixed point before starting a time-consuming iterative process is of high practical relevance to network designers and to algorithms for self-organizing networks.

In the above-mentioned load estimation problem, existence of a fixed point is fully characterized by the spectral radius of a matrix that is easily constructed from the associated concave mapping \cite{siomina12}. We can also use this matrix to build an affine mapping having as its fixed point a vector that gives a lower bound of the network load. The main advantage of working with affine mappings in finite dimensional spaces is that computation of their fixed points reduces to solving simple systems of linear equations, so we may easily obtain in this way a certificate that the current network configuration is not able to serve the demanded traffic. 

The first objective of this study is to show that, by using the concept of recession or asymptotic functions in convex analysis \cite{aus03,baus11}, the technique used in \cite{siomina12} for the construction of the above-mentioned matrices (hereafter called {\it lower bounding matrices}) admits a simple extension to general positive concave mappings. This extension has been motivated by recent results in power estimation in LTE networks \cite{renato14SIP,ho2015}, which deal with mappings different from that considered in \cite{Majewski2010,siomina12}. Concave mappings are also common in many applications in different fields \cite{kennan2001}, so the results of this study are relevant for applications outside of the wireless domain. We show alternative construction methods for lower bounding matrices that are very simple in many applications, including those originally considered in \cite{siomina12}. We also prove that the spectral radius of lower bounding matrices of general concave mappings gives a necessary condition for the existence of fixed points. For some particular concave mappings, this condition is shown to be sufficient. 

The second objective of this study is to develop an acceleration method for the standard fixed-point iteration described in \cite{yates95} when applied to concave mappings. More specifically, we combine the lower bounding matrix and the original positive concave mapping to generate a new mapping that has the same fixed point of the original concave mapping. By applying the standard fixed point iteration to this new mapping, the convergence speed is improved in a well-defined sense, and the computational complexity is not unduly increased because only one additional matrix-vector multiplication per iteration is required. As exemplary applications of the above results, we consider the problems of power and load estimation in OFDMA-based systems \cite{ho2015,renato14SIP}.

This study is structured as follows. In Sect.~\ref{sect.preliminaries} we review basic results in convex analysis and in interference calculus. The material in Sect.~\ref{sect.preliminaries} can now be  considered standard, but we also show a simple proof of the fact that positive concave functions are standard interference functions. In Sect.~\ref{sect.preliminaries} we relate some results in \cite{siomina12} (used to compute lower bounds for load in LTE network planning) to standard results on recession functions in convex analysis. The relations are used in Sect.~\ref{sect.acceleration} to derive conditions for the existence of fixed points of general positive concave mappings. We also propose novel low-complexity iterative methods that improve the convergence speed of the standard fixed point algorithm. In Sect.~\ref{sect.applications} we revisit the problems of load and power estimation in OFDMA-based networks, and we show how the novel results and algorithms proposed here can be used in these concrete applications.

\section{Preliminaries}
\label{sect.preliminaries}

In this study, we use the following standard definitions. By $\innerprod{\signal{x}}{\signal{y}}$ for arbitrary  $\signal{x}\in\real^N$ and $\signal{y}\in\real^N$, we denote the standard inner product $\innerprod{\signal{x}}{\signal{y}}:=\signal{x}^t\signal{y}$. Its induced norm is given by $\|\signal{x}\|:=\sqrt{\innerprod{\signal{x}}{\signal{x}}}$.  The set $\mathcal{I}:=\{1,\ldots,N\}$ is the set of indices of the components of vectors in $\real^N$. The $\infty$-norm of a vector $\signal{x}=[x_1,\ldots,x_N]$ is the norm given by $\|\signal{x}\|_\infty:=\max_{i\in\mathcal{I}} |x_i|$. We define by $\signal{e}_k$ the $k$th standard basis vector of $\real^N$. Vector inequalities should be understood as component-wise inequalities, and we define $\real_{+}^N:=[0,\infty[^N$ and $\real_{++}^N:=~]0,\infty[^N$ (the superscript is omitted if $N=1$). The set of positive integers is denoted by $\Natural:=\{1,~2,\ldots\}$. The spectral radius of a matrix $\signal{M}\in\real^{N\times N}$ is given by $\rho(\signal{M}):=\max\{|\lambda_1|,\ldots,|\lambda_N|\}$, where $\lambda_1,\ldots,\lambda_N$ are the eigenvalues of the matrix $\signal{M}$. The component of the $i$th row and $k$th column of a matrix $\signal{M}$ is denoted by $[\signal{M}]_{i,k}$. For a vector $\signal{x}=[x_1,\ldots,x_N]\in\real^N$, the matrix $\mathrm{diag}(\signal{x})\in\real^{N\times N}$ is a diagonal matrix with $[\mathrm{diag}(\signal{x})]_{i,i}=x_i$.
 
Concave functions and standard interference functions play a crucial role in this study, so we review below basic definitions and known results that are extensively used in the next sections. 

\begin{definition}
(Convex set) A set $C\subset\real^N$ is said to be convex if
\begin{multline*}
(\forall\signal{x}\in C)(\forall\signal{y}\in C)(\forall \alpha\in~]0,1[)\quad \alpha\signal{x}+(1-\alpha)\signal{y}\in C.
\end{multline*}
\end{definition}

\begin{definition}
(Concave functions) We say that $f:C\to\real\cup\{-\infty\}$ is a concave function if $C\subset\real^N$ is a convex set and 
\begin{multline*}
(\forall \signal{x}\in \mathrm{dom}~f)(\forall \signal{y}\in \mathrm{dom}~f)(\forall \alpha\in~]0,1[) \\
f(\alpha\signal{x}+(1-\alpha)\signal{y}) \ge \alpha f(\signal{x})+(1-\alpha)f(\signal{y}),
\end{multline*}
where $\mathrm{dom}~f:=\{\signal{x}\in C~|~f(\signal{x})>-\infty\}$ is the (effective) domain of $f$.~\footnote{In the literature, when a concave function $f$ is allowed to take the value $-\infty$, assuming that $C=\real^N$ is a common practice. If $C$ is a proper subset of $\real^N$, we can define $f(\signal{x})=-\infty$ if $\signal{x}\in\real^N\backslash C$ to extend $f$ from $C$ to $\real^N$. By doing so, the effective domain is preserved. However, for notational convenience later in the text, we do not necessarily adhere to this convention, and we allow $C$ to be a strict subset of $\real^N$.}
\end{definition}

Concave functions $f:C\to\real\cup\{-\infty\}$ with $C\subset\real^N$ for which there exists at least one vector $\signal{x}\in C$ satisfying $f(\signal{x})>-\infty$ are called {\it proper} concave functions. If for every sequence $\{\signal{x}_n\}_{n\in\Natural}\subset C$ converging to an arbitrary vector $\signal{x}\in C$, we have $\limsup_{n\to\infty}f(\signal{x}_n)\le f(\signal{x})$, then we say that the function $f$ is upper semicontinuous (on $C$).

Every concave function $f:C\to\real\cup\{-\infty\}$ can be related to a convex function by $-f:C\to\real\cup\{\infty\}$ ($-f$ takes the value $\infty$ whenever $f$ takes the value $-\infty$), so the following results on concave functions can be directly deduced from standard results on convex functions found in the literature \cite{baus11,aus03}.

\begin{definition}
\label{sup.inequality}
(Superdifferentials and supergradients) Let $f:~C\to \real\cup \{-\infty\}$ be a concave function with $\emptyset\ne C\subset\real^N$. The superdifferential of $f$ at $\signal{x}\in \mathrm{dom}~f$ is the set given by
 \begin{align*}
 \partial f(\signal{x}) := \left\{\signal{u}\in\real^N~|~(\forall \signal{y}\in C)~\innerprod{\signal{y}-\signal{x}}{\signal{u}}+f(\signal{x})\ge f(\signal{y})\right\}.
 \end{align*}
 If $\signal{x}\notin\mathrm{dom}~f$, then we define $\partial f(\signal{x}):=\emptyset$. A vector $\signal{g}\in \partial f(\signal{x})$ is called a supergradient of $f$ at $\signal{x}$. The domain of the superdifferential $\partial f$ is the set given by $\mathrm{dom}~\partial f:=\{\signal{x}\in\real^N~|~\partial f(\signal{x})\neq\emptyset\}$.
\end{definition}

In this study, if the point at which a supergradient is selected needs to be explicitly known, then we often use the notation $\signal{g}(\signal{x})\in\partial f(\signal{x})$ to denote an arbitrary choice of the supergradient at $\signal{x}$.

 As a particular case of \cite[Corollary~16.15]{baus11}, we have the following:

\begin{fact}
\label{fact.existence}
Let $f:\real^N_{+}\to\real_{++}$ be concave. Then the superdifferential $\partial f(\signal{x})$ is nonempty for every $\signal{x} \in\real_{++}^N$.
\end{fact}

The proposed acceleration methods are based on the concept of recession functions (or asymptotic functions) in convex analysis.

\begin{definition}
\label{definition.recession}
(Recession or asymptotic functions) Let $f:\real^N\to\real\cup\{-\infty\}$ be upper semicontinuous, proper, and concave. We define as its recession or asymptotic function at $\signal{y}\in\real^N$ the function given by (see \cite[Ch. 2.5]{aus03}\cite[p. 152]{baus11} for the standard definition for convex $f$):
\begin{align*}
(\forall \signal{x}\in\mathrm{dom} f)~~ f_\infty(\signal{y}):=\lim_{h\to\infty}\dfrac{f(\signal{x}+h\signal{y})-f(\signal{x})}{h}.
\end{align*}
(NOTE: The above limit is always well defined. We assume that it can take the value $-\infty$.)
\end{definition}

\begin{fact}
\label{fact.recession_func}
If $f:\real^N\to\real\cup\{-\infty\}$ is a proper, upper semicontinuous, and concave function, then for every $\signal{y}\in \mathrm{dom} f$  we have \cite[Corollary 2.5.3]{aus03}
\begin{align}
\label{eq.recession_func}
f_\infty(\signal{y})=\lim_{h\to 0^+}h f(h^{-1}\signal{y}),
\end{align}
and the above is valid for every $\signal{y}\in\real^N$ if $\signal{0}\in\mathrm{dom}~f$.
\end{fact}

\begin{fact}
\label{fact.sub_recession}
Let $f:\real^N\to\real\cup\{-\infty\}$ be proper, upper semicontinuous, and concave. Then \cite[Proposition 6.5.1]{aus03}
\begin{align*}
f_\infty(\signal{y})=\inf\{\innerprod{\signal{g}}{\signal{y}}~|~\signal{x}\in\mathrm{dom}~\partial f,~\signal{g}\in~\partial f(\signal{x})\}.
\end{align*}
\end{fact}

Many estimation and optimization tasks in communication networks can often be posed as systems coupled by standard interference functions, which we define below.

\begin{definition}
\label{definition.inter_func}
 (Standard interference functions and mappings \cite{yates95}) A function $f: \real_+^N \to \real_{++}$ is said to be a standard interference function if the following properties hold:\par 
\begin{enumerate}
 \item ({\it Scalability}) $(\forall \signal{x}\in\real^N_+)$ $(\forall \alpha>1)$  $\alpha {f}(\signal{x})>f(\alpha\signal{x})$. \par
 \item ({\it Monotonicity}) $(\forall \signal{x}_1\in\real_+^N)$ $(\forall \signal{x}_2\in\real_+^N)$ $\signal{x}_1\ge\signal{x}_2 \Rightarrow{f}(\signal{x}_1)\ge f(\signal{x}_2)$. \par 
\end{enumerate}
Given $N$ standard interference functions $f_i:\real^N_+\to\real_{++}$, $i=1,\ldots,N$, we call the mapping $T:\real^N_+\to\real_{++}^N:\signal{x}\mapsto[f_1(\signal{x}),\ldots, f_N(\signal{x})]$ a {\it standard interference mapping}.
\end{definition}

\begin{fact} \label{fact.inter_map} (Properties of interference mappings \cite{yates95}) Let ${T}:\real^N_{+}\to\real^N_{++}$ be a standard interference mapping. Then the following holds:
\begin{subfact}\label{fact.uniqueness} Let $\mathrm{Fix}(T):=\{\signal{x}\in\real^N_{++}~|~T(\signal{x})=\signal{x}\}$ be the set of fixed points of $T$, then $\mathrm{Fix}(T)$ is either an empty set or a singleton.
\end{subfact}
\begin{subfact}
\label{subfact.upper_bound}
 $\mathrm{Fix}(T)\neq\emptyset$ if and only if there exists $\signal{x}^\prime\in\real^N$ such that ${T}(\signal{x}^\prime)\le\signal{x}^\prime$.
\end{subfact}
\begin{subfact} 
\label{subfact.monotone}
If $\mathrm{Fix}(T)\neq\emptyset$, then the fixed point of $T$ is the limit of the sequence $\{\signal{x}_n\}_{n\in\Natural}$ generated by $\signal{x}_{n+1}={T}(\signal{x}_n)$, where $\signal{x}_1\in\real^N_{+}$ is arbitrary.\footnote{In finite dimensional spaces, all norms are equivalent. Therefore, convergence of a sequence $\{\signal{x}_n\}_{n\in\Natural}\subset \real^N$ to a point  $\signal{x}^\star\in\real^N$ does not depend on the choice of the norm; i.e., $\lim_{n\to\infty}\|\signal{x}_n-\signal{x}^\star\|_\mathrm{a}=0$ for any norm $\|\cdot\|_\mathrm{a}$.} If $\signal{x}_1$ satisfies ${T}(\signal{x}_1)\ge \signal{x}_1$ (resp. ${T}(\signal{x}_1)\le \signal{x}_1$), then the sequence is monotonically increasing (resp. monotonically decreasing) in each component. In particular, monotonically increasing sequences are produced with $\signal{x}_1=\signal{0}$.
\end{subfact}
\end{fact}

The focus of this study is on (positive) concave functions, which as shown below are a subclass of standard interference functions.

\begin{proposition}
\label{prop.pos_conc}
 Concave functions $f:\real^N_{+}\to\real_{++}$ are standard interference functions:
\end{proposition}
\begin{proof}
We need to prove that concave functions $f:\real^N_{+}\to\real_{++}$ satisfy the scalability and monotonicity properties in Definition~\ref{definition.inter_func}.\par 
({Scalability}) Let $\mu>1$ and $\signal{x}\in\real_+^N$ be arbitrary. By concavity of $f$, for every $\alpha\in~]0,1[$, we have $f(\alpha\mu\signal{x})=f(\alpha\mu\signal{x}+(1-\alpha)\signal{0})\ge \alpha f(\mu\signal{x})+(1-\alpha)f(\signal{0})$. In particular, for $\alpha=1/\mu$, we conclude from the last inequality and positivity of $f$ that
\begin{align*}
f(\signal{x})\ge \dfrac{1}{\mu} f(\mu\signal{x}) + \left(1-\dfrac{1}{\mu}\right)  f(\signal{0}) > \dfrac{1}{\mu} f(\mu\signal{x}),
\end{align*}
which proves the scalability property. \par 

(Monotonicity) Let $(\signal{x}_1,\signal{x}_2)\in\real_{+}^N\times \real_{+}^N$ satisfy $\signal{x}_2\ge \signal{x}_1$. As a result, $\signal{x}_1+\mu(\signal{x}_2-\signal{x}_1)\in\real_+^{N}$ for every $\mu\ge 0$. From the definition of concavity, we also have
\begin{multline*}
(\forall \alpha\in~]0,1[)(\forall\mu\ge 0) \\ f\left((1-\alpha)\signal{x}_1 + \alpha \left(\signal{x}_1+\mu(\signal{x}_2-\signal{x}_1) \right)\right) \ge \\ 
 (1-\alpha) f(\signal{x}_1) + \alpha f\left(\signal{x}_1+\mu(\signal{x}_2-\signal{x}_1)\right).
\end{multline*}
In particular, for an arbitrary $\mu>1$  and for $\alpha = 1/\mu$, we obtain from the positivity of $f$ that
\begin{align*}
f(\signal{x}_2) \ge \left(1-\dfrac{1}{\mu}\right) f(\signal{x}_1)+\dfrac{1}{\mu}f\left(\signal{x}_1+\mu(\signal{x}_2-\signal{x}_1)\right) \\ > f(\signal{x}_1)-\dfrac{1}{\mu} f(\signal{x}_1).
\end{align*}

The inequality $f(\signal{x}_2)>f(\signal{x}_1)-({1}/{\mu}) f(\signal{x}_1)$ is valid for every $\mu>1$, so we can take the limit as $\mu$ goes to infinity to conclude that
\begin{align*}
f(\signal{x}_2)\ge \lim_{\mu\to\infty} \left( f(\signal{x}_1) - \dfrac{1}{\mu} f(\signal{x}_1)\right) = {f}(\signal{x}_1).
\end{align*}

\end{proof}

As every result stated in this section, Proposition~\ref{prop.pos_conc} can be considered standard (see \cite{renato14SPM} and the references therein). Nevertheless, we have decided to include a simple proof of this proposition because similar statements can often be found in the literature without proof. Furthermore, some partial proofs available in the literature make implicit assumptions such as the existence of the supergradients on the boundary of the domain $\real_{+}^N$ and/or the strict concavity of the functions. We emphasize that these assumptions are not required. {As an example of a positive concave function (and hence a standard interference function) not satisfying these two assumptions, we have
\begin{align*}
f:\real_+\to\real_{++}:x\mapsto \begin{cases}
1,&\mathrm{if}~~x=0, \\
2,&\mathrm{otherwise}.
\end{cases}
\end{align*}

To characterize the existence of fixed points of affine standard interference mappings, we can use the following fact: 

\begin{fact}
\label{fact.neuman} \cite[Theorem A.16]{slawomir09}
For an arbitrary matrix $\signal{M}\in\real^{N\times N}$, if $\rho(\signal{M})<1$, then $\sum_{k=1}^\infty \signal{M}^k$ converges and $(\signal{I}-\signal{M})^{-1}=\signal{I}+\sum_{k=1}^\infty \signal{M}^k$.
\end{fact}

\begin{fact}
\label{fact.spec_radius}
\cite[Theorem A.51]{slawomir09} Let $\signal{M}\in\real_{+}^{N\times N}$ be a non-negative matrix, and let $\signal{p}\in\real^N_{++}$ be arbitrary. A sufficient and necessary condition for the system $\signal{x}=\signal{p}+\signal{Mx}$ to have a (strictly) positive solution $\signal{x}\in\real_{++}^N$ is $\rho(\signal{M})< 1$.
\end{fact}

We end this section with a very simple statement that is used later to clarify an argument in Sect.~\ref{sect.applications}.

\begin{remark}
\label{remark.eigenvalues} 
Let $\signal{M}\in\real^{N\times N}$ be arbitrary and $\signal{D}\in\real^{N\times N}$ be an invertible matrix. Then the eigenvalues of the matrices $\signal{M}$ and $\signal{D}\signal{M}\signal{D}^{-1}$ are the same (which in particular implies that $\rho(\signal{M})=\rho(\signal{D}\signal{M}\signal{D}^{-1})$).
\end{remark}
\begin{proof}
 Assume that $\signal{x}$ is a right eigenvector associated with an eigenvalue $\lambda$, and define $\signal{y}:=\signal{D}\signal{x}$. As a consequence,
\begin{align*}
\signal{M}\signal{x}=\lambda \signal{x} \Leftrightarrow \signal{M}\signal{D}^{-1}\signal{y}=\lambda \signal{D}^{-1} \signal{y} \Leftrightarrow \signal{D}\signal{M}\signal{D}^{-1}\signal{y}=\lambda \signal{y},
\end{align*}
and the result follows.

\end{proof}

\section{Component-wise infimum of supergradients of positive concave functions}
\label{sect.infimum}

The main objective of this section is to propose two simple techniques for computing the component-wise infimum of supergradients of concave functions (c.f. Proposition~\ref{proposition.inf_1} and Proposition~\ref{proposition.inf_2}). These techniques are motivated by the following application. In load estimation problems in wireless networks, the values taken by partial derivatives of functions related to the load coupling among base stations attain their infimum asymptotically as we move to infinity in the direction of a basis vector \cite{Majewski2010,siomina2012load}. This observation has given rise to efficient techniques for the computation of lower bounds for the load in that very particular application domain \cite{siomina2012load}, and extending these results to a more general class of concave functions is highly desirable for other applications such as power estimation in networks. 

By using the concept of recession or asymptotic functions, we show below that the above-mentioned asymptotic result can be generalized to all positive concave functions, even if the functions are not differentiable, in which case we use supergradients instead of gradients (c.f. Proposition~\ref{proposition.inf_2}). We can further show that the component-wise infimum taken by the supergradients can be easily obtained by means of simple schemes that do not require the computation of supergradients (c.f. Proposition~\ref{proposition.inf_1}). These infimum values are used later by the proposed acceleration schemes to compute fixed points of positive concave mappings, and they can also be used to obtain a certificate that the mapping does not have a fixed point. We start by formalizing some simple properties of supergradients of concave functions.

\begin{lemma} Let $f:\real^N_{+}\to\real_{++}$ be an upper semicontinuous concave function. Then the following holds:

\begin{sublemma} \label{slemma.positivity} All supergradients of $f$ are non-negative vectors; i.e.,
\begin{equation*}
\left(\forall \signal{x}\in\mathrm{dom}~\partial f\right)\left(\forall \signal{g} \in \partial f(\signal{x})\right)\quad \signal{g}\ge \signal{0}.
\end{equation*}
\end{sublemma}

\begin{sublemma}
 \label{slemma.nondecreasing}
 Let $\signal{x}\in\real^N_+$ and $k\in\mathcal{I}$ be arbitrary and assume that $\signal{x}+h\signal{e}_k\in\mathrm{dom}~ \partial f$ for every $h\ge 0$.  Then \\
   $(\forall h>0)$ $\left(\forall \signal{g}^\prime\in\partial f(\signal{x})\right)$ $\left(\forall \signal{g}^\dprime\in\partial f(\signal{x}+h \signal{e}_k)\right)$
\begin{align*}
0\le{g}^\dprime_k \le {g}^\prime_k,
\end{align*}
where $[g_{1}^\prime,\ldots,g_N^\prime]^t:=\signal{g}^\prime$ and $[g_1^\dprime,\ldots,g_N^\dprime]^t:=\signal{g}^\dprime$.
\end{sublemma}

\begin{sublemma} \label{slemma.second_ineq} As in Lemma~\ref{slemma.nondecreasing}, let $\signal{x}\in\real^N_+$ and $k\in\mathcal{I}$ be arbitrary and assume that $\signal{x}+h\signal{e}_k\in\mathrm{dom}~ \partial f$ for every $h\ge 0$.  Then
\begin{multline}
\label{eq.second_ineq}
(\forall h>0)(\forall \signal{g}(\signal{x}+h\signal{e}_k) \in \partial f(\signal{x}+h\signal{e}_k))\\
g_k(\signal{x}+h\signal{e}_k)\le\dfrac{f(\signal{x}+h\signal{e}_k)-f(\signal{x})}{h},
\end{multline}
where $[g_1(\signal{x}+h\signal{e}_k),\dots,g_N(\signal{x}+h\signal{e}_k)]^t:=\signal{g}(\signal{x}+h\signal{e}_k)$.
\end{sublemma}

\end{lemma}
\begin{proof}

\begin{enumerate}
\item We prove the result by contradiction. Assume that there exists a supergradient $\signal{g}=:[g_1,\ldots,g_N]^t\in\partial f(\signal{x}^\prime)$ at some point $\signal{x}^\prime\in\mathrm{dom}~\partial f$ such that $g_i<0$ for an arbitrary $i\in\{1,\ldots,N\}$. We know from the definition of supergradients that 
\begin{align*}
(\forall \signal{y}\in\real^N_{+})~ f(\signal{y}) \le f(\signal{x}^\prime) + \signal{g}^t(\signal{y}-\signal{x}^\prime).
\end{align*}
In particular, for $\signal{u}:\real\to\real^N:h\mapsto \signal{x}^\prime+h\signal{e}_i$, we obtain
\begin{align*}
f(\signal{u}(h)) & \le f(\signal{x}^\prime) + \signal{g}^t(\signal{u}(h)-\signal{x}^\prime) = f(\signal{x}^\prime)+g_i h.
\end{align*}
Now, since $g_i<0$ by assumption,  we obtain $f(\signal{u}(h))\le 0$ for an arbitrary $h \ge f(\signal{x}^\prime)/|g_i|$, which contradicts the positivity of the range of the function $f:\real^N_{+}\to\real_{++}$. This proves Lemma~\ref{slemma.positivity}.

\item By Definition~\ref{sup.inequality}, for arbitrary $\signal{x}_1,\signal{x}_2\in\mathrm{dom}~\partial f$, we have $f(\signal{x}_1)\le f(\signal{x}_0)+\signal{g}_0^t(\signal{x}_1-\signal{x}_0)$ and 
$f(\signal{x}_0)\le f(\signal{x}_1)+\signal{g}_1^t(\signal{x}_0-\signal{x}_1)$, where $\signal{g}_0\in\partial f(\signal{x}_0)$ and $\signal{g}_1\in\partial f(\signal{x}_1)$ are arbitrary supergradients. Summing these two inequalities yields
\begin{align}
\label{eq.monotone}
\left(\signal{g}_1-\signal{g}_0\right)^t(\signal{x}_1-\signal{x}_0)\le 0.
\end{align}

In particular, for $\signal{x}_1=\signal{x}+h\signal{e}_k$ and $\signal{x}_0=\signal{x}$, we have $\signal{0}\le\signal{x}_1-\signal{x}_0=h\signal{e}_k\neq\signal{0}$, and we can set $\signal{g}_0=\signal{g}^\prime\in\partial f(\signal{x})$ and $\signal{g}_1=\signal{g}^\dprime\in\partial f(\signal{x}+h\signal{e}_k)$. Using these particular choices for $\signal{x}_0$, $\signal{x}_1$, $\signal{g}_0$, and $\signal{g}_1$ in \refeq{eq.monotone}, we obtain $g_k^\dprime \le g_k^\prime$. Non-negativity of ${g}_k^\dprime$ has been proved in the first part of the lemma.

\item Use $\signal{x}^\prime=\signal{x}$ in the supergradient inequality ${f}(\signal{x}^\prime)\le f(\signal{x}+h\signal{e}_k)+\signal{g}(\signal{x}+h\signal{e}_k)^t(\signal{x}^\prime-\signal{x}-h\signal{e}_k)$.

\end{enumerate}

\end{proof}

We can now show an efficient scheme to compute the element-wise infimum of supergradients. 

\begin{proposition}
\label{proposition.inf_1}
Let $S:=\bigcup_{\signal{x}\in\real^N_{+}} \partial f(\signal{x})$ be the set of all supergradients of an upper semicontinuous concave function $f:\real^N_{+}\to\real_{++}$. For each $k\in\mathcal{I}$, define 
\begin{align}
\label{eq.infimum}
g_k^\star:=\inf\left\{g_k\in\real~|~ \signal{g}=[g_1,\dots,g_N] \in S\right\}\in\real_{+},
\end{align}
 then we have 
\begin{align*}
(\forall k\in\mathcal{I})~~ g_k^\star=\lim_{h\to 0^+} {h f(h^{-1}\signal{e}_k)}.
\end{align*}
\end{proposition}
\begin{proof}

We have $g_k^\star\ge 0$ as a direct consequence of Lemma~\ref{slemma.positivity}. Now consider the standard extension $\tilde{f}:\real^N\to\real_{++}\cup\{-\infty\}$ of $f:\real_+^N\to\real_{++}$ given by
\begin{align*}
\tilde{f}(\signal{x})=\begin{cases}
f(\signal{x}) &\mathrm{if}~~\signal{x}\in\real^N_+\\ 
-\infty & \mathrm{otherwise,}
\end{cases}
\end{align*}
and let $k\in\mathcal{I}$ be arbitrary. By construction, $\tilde{f}$ is upper semicontinuous, proper, and concave. Furthermore, $\mathrm{dom}~\partial f = \mathrm{dom}~\partial \tilde{f}$, $\mathrm{dom}~ f = \mathrm{dom}  \tilde{f}$, and $\partial f(\signal{x})=\partial \tilde{f}(\signal{x})$ for every $\signal{x}\in \mathrm{dom} ~\partial\tilde{f}$. By Fact~\ref{fact.sub_recession}, we have $g_k^\star=\tilde{f}_\infty(\signal{e}_k)$ (see Definition~\ref{definition.recession}), and the result now follows from Fact~\ref{fact.recession_func}.
\end{proof}

Next, we show an alternative means of computing $g_k^\star$ in \refeq{eq.infimum}. This alternative method has been used in \cite{siomina12} for a very particular concave function appearing in load estimation in LTE networks (see Sect.~\ref{sect.load_planning}).

\begin{proposition}
\label{proposition.inf_2}
Let $\signal{x}\in\real^N_{+}$ and $k\in\mathcal{I}$ be arbitrary. In addition, assume that $f:\real_+^N\to\real_{++}$ is an upper semicontinuous concave function and that $\signal{x}+h\signal{e}_k\in\mathrm{dom}~ \partial f$ for every $h\ge 0$. Define by $[{g_1}(\signal{x}+h\signal{e}_k),\ldots,{g_N}(\signal{x}+h\signal{e}_k)]:= \signal{g}(\signal{x}+h\signal{e}_k)\in\partial f(\signal{x}+h\signal{e}_k)$ an arbitrary supergradient at $\signal{x}+h\signal{e}_k$. Then 
\begin{align}
 \lim_{h\to\infty} {g_k}(\signal{x}+h\signal{e}_k) = g_k^\star\ge 0,
\end{align}
where $g_k^\star$ is defined in \refeq{eq.infimum}.
\end{proposition}
\begin{proof}
Let $k\in\mathcal{I}$ be arbitrary. It follows from Lemma~\ref{slemma.nondecreasing} that, irrespective of the criterion we use to select a supergradient $\signal{g}(\signal{x}+h\signal{e}_k)\in\partial f(\signal{x}+h\signal{e}_k)$, its $k$th component $g_k(\signal{x}+h\signal{e}_k)$ should be monotonically non-increasing as $h$ increases (and lower bounded by 0). As a result, the limit $\lim_{h\to\infty} g_k(\signal{x}+h\signal{e}_k)$ exists. By definition, $g_k^\star$ is the infimum of the $k$th component of all supergradients, hence we have that 
\begin{align*} 
g_k^\star\le \lim_{h\to\infty} g_k(\signal{x}+h\signal{e}_k) 
\end{align*}
for any choice of $\signal{x}$ and $k$ satisfying the assumptions of the lemma. Using \refeq{eq.second_ineq} in Lemma~\ref{slemma.second_ineq} and the definition of recession functions, we deduce

\begin{multline*}
g_k^\star\le \lim_{h\to\infty} g_k(\signal{x}+h\signal{e}_k) \\ \le \lim_{h\to\infty}\dfrac{f(\signal{x}+h\signal{e}_k)-f(\signal{x})}{h}  = f_\infty(\signal{e}_k).
\end{multline*}
The result now follows by noticing that $f_\infty(\signal{e}_k)=g_k^\star$ by Fact~\ref{fact.sub_recession}. (Non-negativity of $g_k^\star$ is immediate from Lemma~\ref{slemma.positivity}.)
\end{proof}

\section{Acceleration algorithms for positive concave mappings}
\label{sect.acceleration}

Having two efficient methods to compute the component-wise infimum of supergradients of concave functions, we can now proceed with the study of general concave mappings. To avoid unnecessary technical digressions, we do not deal with concave functions $f:\real_+^N\to\real_{++}$ that are not upper semicontinuous. To formalize this assumption, we use the following definition:

\begin{definition}
(Positive concave mappings) We say that $T:\real_+^N\to \real_{++}^N$ is a positive concave mapping if it is given by
\begin{align}
\label{eq.mapping}
T(\signal{x}):=[f_1(\signal{x}),\ldots,f_N(\signal{x})]^t,
\end{align}
where all functions $f_1:\real^N_{+}\to\real_{++},\ldots, f_N:\real^N_{+}\to\real_{++}$ are concave and upper semicontinuous.
\end{definition}

By Proposition~\ref{prop.pos_conc}, we know that positive concave mappings are standard interference mappings. The remaining of this section has the objective of investigating the following problems associated with a positive concave mapping $T$ (which, as shown in Sect.~\ref{sect.applications}, are problems that need to be addressed in many network planning and optimization tasks):
\begin{itemize}

\item[P1)] Verify whether $T$ has a fixed point by using computationally efficient algorithms. 

\item[P2)] Improve the convergence speed of the standard iteration in Fact~\ref{subfact.monotone} to obtain the fixed point of $T$ (if it exists). 
\end{itemize}

\subsection{Conditions for the existence of fixed points of positive concave mappings}
To address problem P1), we use the concept of lower bounding matrices, which we define as follows:
\begin{definition}
\label{definition.lower_bounding}
 The lower bounding matrix of a positive concave mapping $T:\real_+^N\to\real_{++}^N:\signal{x}\mapsto [f_1(\signal{x}),\ldots,f_N(\signal{x})]^t$ is the non-negative matrix $\signal{M}\in\real_+^{N\times N}$ with its $i$th row and $k$th column given by
\begin{multline}
[\signal{M}]_{i,k}:=\inf\left\{g_k\in\real~|~ [g_1,\dots,g_N] \in S_i\right\}\in\real_{+},
\end{multline}
where $S_i:=\bigcup_{\signal{x}\in\real_+^N} \partial f_i(\signal{x})$. 
\end{definition}

Note that Proposition~\ref{proposition.inf_1} and Proposition~\ref{proposition.inf_2} show two simple techniques to compute each component of lower bounding matrices.

\begin{example}
\label{example.derivation}
(Construction of lower bounding matrices with the results in Propositions~\ref{proposition.inf_1} and \ref{proposition.inf_2}) Let the functions $f_1:\real^N_{+}\to\real_{++},\ldots, f_N:\real^N_{+}\to\real_{++}$ be concave and upper semicontinuous. Using Proposition~\ref{proposition.inf_1}, we can compute the lower bounding matrix $\signal{M}$ of the mapping $T:\real^N_+\to\real^N_{++}:\signal{x} \mapsto [f_1(\signal{x}),\ldots,f_N(\signal{x})]^t$ by
 \begin{multline}
 \label{eq.derivation_rec}
 \signal{M} = \\ \left[\begin{matrix} 
 \lim_{h\to 0^+} h {f}_{1}(h^{-1}\signal{e}_1) & \cdots & \lim_{h\to 0^+} h{f}_{1}(h^{-1}\signal{e}_N) \\
 \vdots & \ddots & \vdots \\
  \lim_{h\to 0^+} h {f}_{N}(h^{-1}\signal{e}_1) & \cdots & \lim_{h\to 0^+} h{f}_{N}(h^{-1}\signal{e}_N) 
 \end{matrix}\right].
 \end{multline}
Equivalently, by fixing $\signal{x}^\prime\in\real_{++}^N$ arbitrarily, we can also compute the lower bounding matrix $\signal{M}$ with the results in Proposition~\ref{proposition.inf_2} and Fact~\ref{fact.existence} as follows:
\begin{multline}
 \label{eq.derivation_supgrad}
\signal{M} = \\ \left[\begin{matrix} 
\lim_{h\to\infty} g_1^1(\signal{x}^\prime+h\signal{e}_1) & \cdots & \lim_{h\to\infty} g_N^1(\signal{x}^\prime+h\signal{e}_N) \\
\vdots & \ddots & \vdots \\
\lim_{h\to\infty} g_1^N(\signal{x}^\prime+h\signal{e}_1) & \cdots & \lim_{h\to\infty} g_N^N(\signal{x}^\prime+h\signal{e}_N) \\
\end{matrix}\right],
\end{multline}
where we denote by $g_k^i(\signal{x})$ the $k$th element of  a supergradient of $f_i$ at $\signal{x}\in\real^N_{++}$ (i.e., $[g^i_1(\signal{x}), \ldots, g^i_N(\signal{x})]^t\in\partial f_i(\signal{x})$).
\end{example}

 Proposition~\ref{proposition.inf_1} and Proposition~\ref{proposition.inf_2} also show that the lower bounding matrix is non-negative. The name ``lower bounding matrix'' stems from the fact that this matrix is constructed with component-wise lower bounds of supergradients. Lower bounding matrices can also be used to construct affine mappings that serve as lower bounds of their corresponding positive concave mappings, in the following sense:

\begin{lemma}
\label{lemma.lb}
Let $\signal{M}$ be the lower bounding matrix of a positive concave mapping $T:\real^N_+\to\real^N_{++}$ in accordance with Definition~\ref{definition.lower_bounding}. Then
\begin{align}
\label{eq.lb}
(\forall \signal{y}\in\real_{+}^N) (\forall \signal{x}\in\real_{+}^N) ~ \signal{x}\ge\signal{y}\Rightarrow T(\signal{x})\ge T(\signal{y})+\signal{M}(\signal{x}-\signal{y}).
\end{align}
\end{lemma}
\begin{proof}
We prove the inequality for an arbitrary component of the mapping $T$; i.e., for the function $f_i$, where $i\in\mathcal{I}$ is arbitrary. Let $\signal{y}\in\real^N_+$ and  $\signal{x}\ge\signal{y}$ be arbitrary vectors, and construct the sequence  $\{\signal{x}_n:=(1/n)\signal{1}+\signal{x}\}_{n\in\Natural}\subset\real_{++}^N$. By Fact~\ref{fact.existence}, we have $\signal{x}_n\in\mathrm{dom}~\partial f_i$ for every $n\in\Natural$. From the definition of supergradients, we know that, for every $n\in\Natural$,
\begin{align}
\label{eq.basic_inequality}
f_i(\signal{y})+\signal{g}^t_n(\signal{x}_n-\signal{y})\le f_i(\signal{x}_n),
\end{align}
where $\signal{g}_n\in\partial f_i(\signal{x}_n)$ is an arbitrary supergradient.  By Definition~\ref{definition.lower_bounding}, the $i$th row of $\signal{M}$, denoted by $\signal{m}_i \ge \signal{0}$ as a column vector, is the component-wise infimum of all supergradients of the function $f_i$, hence $\signal{0}\le\signal{m}_i\le\signal{g}_n$ for every $n\in\Natural$. Using this last relation together with $\signal{x}_n\ge\signal{y}$ in \refeq{eq.basic_inequality}, we deduce:
\begin{multline*}
(\forall n\in\Natural) \\
f_i(\signal{y})+\signal{m}^t_i(\signal{x}_n-\signal{y})\le f_i(\signal{y})+\signal{g}^t_n(\signal{x}_n-\signal{y})\le f_i(\signal{x}_n).
\end{multline*}
 By construction, $\lim_{n\to\infty} \signal{x}_n=\signal{x}$. As a result, we conclude from the continuity of affine functions and upper semi-continuity of $f_i$ that
\begin{multline*}
f_i(\signal{y})+\signal{m}^t_i(\signal{x}-\signal{y})=\limsup_{n\to\infty} (f_i(\signal{y})+\signal{m}^t_i(\signal{x}_n-\signal{y}))   \\ \le \limsup_{n\to\infty} f_i(\signal{x}_n) \le  f_i(\signal{x}).
\end{multline*}
\end{proof}

The next proposition addresses problem P1) stated in the beginning of this section:
\begin{proposition}
\label{proposition.spec_radius}
Let $\signal{M}$ be the lower bounding matrix of a positive concave mapping $T:\real_+^N\to\real_{++}^N$. A necessary condition for $\mathrm{Fix}(T)\neq \emptyset$ is $\rho(\signal{M})<1$.

\end{proposition}
\begin{proof}
Use $\signal{y}=\signal{0}$ in \refeq{eq.lb} to verify that the affine mapping $\tl:\real_+^N\to\real_{++}^N:\signal{x}\mapsto T(\signal{0})+\signal{M}\signal{x}$ satisfies $T(\signal{x})\ge \tl(\signal{x})>\signal{0}$ for every $\signal{x}\in\real_{+}^N$. Being an affine mapping, $T_L$ is a positive concave mapping, hence it is also a standard interference mapping by Proposition~\ref{prop.pos_conc}. Now let $\signal{x}^\star\in\real_{++}^N$ be the fixed point of the mapping $T$. By Lemma~\ref{lemma.lb}, we obtain:
\begin{align}
\label{eq.existence}
\tl(\signal{x}^\star)\le T(\signal{x}^\star)=\signal{x}^\star,
\end{align}
which implies the existence of the (unique) fixed point of the mapping $\tl$ by Fact~\ref{fact.inter_map}.2. In other words, there exists a unique positive vector $\widehat{\signal{x}}\in\real^N_{++}$ satisfying $\widehat{\signal{x}} = T(\signal{0})+\signal{M}\widehat{\signal{x}}$, and we know by Fact~\ref{fact.spec_radius} that there exists a positive vector satisfying this last equality if and only if $\rho(\signal{M})< 1$ (recall that $T(\signal{0})>\signal{0}$ and that $\signal{M}\in\real_+^{N\times N}$ by construction). 
\end{proof}

An immediate consequence of Fact~\ref{fact.neuman} and Proposition~\ref{proposition.spec_radius} is the following useful result:

\begin{Cor}
\label{cor.matrix_inverse}
Let $T:\real_{+}^N\to\real_{++}^N$ be a positive concave mapping with $\mathrm{Fix}(T)\neq\emptyset$, and denote by $\signal{M}\in\real_+^{N \times N}$ its existing lower bounding matrix given by Definition~\ref{definition.lower_bounding}. Then $(\signal{I}-\signal{M})^{-1}$ exists, and it is a non-negative matrix.
\end{Cor}

Proposition~\ref{proposition.spec_radius} is interesting in its own right because it enables us to certify that a given positive concave mapping has no fixed point. We only need to show that the spectral radius of its lower bounding matrix has spectral radius greater than or equal to one. This result is highly relevant in network optimization and planning problems. As already mentioned in the introduction, in these applications, the feasibility of a network design follows from the existence of the fixed point of a mapping that is constructed based on antenna tilts, power allocations, the position of base stations, etc. Optimization of the network performance (e.g., in terms of energy efficiency, capacity, coverage, etc.) over the joint set of all control parameters is typically an NP-hard problem. As a result, many optimization algorithms proposed in the literature are greedy heuristics that need a fast feasibility check of multiple network configurations at each iteration \cite{Majewski2010}. Proposition~\ref{proposition.spec_radius} opens up the door to the development of efficient and fast methods for excluding many infeasible network configurations from consideration, which can significantly accelerate the overall optimization process.

We emphasize that the converse of Proposition~\ref{proposition.spec_radius} does not hold in general. There are mappings for which the lower bounding matrix has spectral radius strictly less than one, and yet mappings do not have a fixed point (see the application in Sect.~\ref{sect.power_planning}). Therefore, to characterize the existence of a fixed point based on the spectral radius of the lower bounding matrix, we need additional assumptions on the mapping. The next proposition shows a particularly useful assumption that is satisfied in load estimation problems (see Sect.~\ref{sect.load_planning} and \cite{siomina12} for a particular application of this proposition).

\begin{proposition}
\label{prop.sufficiency}
Let ${T}:\real_{+}^N\to\real_{++}^N$ be a positive concave mapping with lower bounding matrix $\signal{M}$ satisfying $\rho(\signal{M})<1$. In addition, assume that 
\begin{align}
\label{eq.upper_value}
(\exists \signal{y}\in\real_{++}^N)(\forall \signal{x}\in\real_{+}^N)\quad  {T}(\signal{x})\le \signal{y}+\signal{M}\signal{x}.
\end{align}
Then the mapping $T$ has a fixed point.
\end{proposition}
\begin{proof}
Let $\signal{y}^\prime\in\real_{++}^N$ be a vector satisfying ${T}(\signal{x})\le \signal{y}^\prime+\signal{M}\signal{x}$ for every $\signal{x}\in\real_+^N$. By Fact.~\ref{fact.spec_radius}, we know that $\signal{x}^\prime:=(\signal{I}-\signal{M})^{-1}\signal{y}^\prime$ is a strictly positive vector. Therefore,
\begin{align*}
T(\signal{x}^\prime)\le \signal{y}^\prime + \signal{M}\signal{x}^\prime = (\signal{I}-\signal{M})\signal{x}^\prime + \signal{M}\signal{x}^\prime = \signal{x}^\prime,
\end{align*}
and the above implies that $\mathrm{Fix}(T)\ne\emptyset$ by Fact.~\ref{subfact.upper_bound}. \par

\end{proof}

\subsection{Acceleration techniques for positive concave mappings}
\label{sect.algorithm}

We now turn our attention to problem P2). To address this problem, we use the concept of accelerated mappings, which we define as follows:
\begin{definition}
(Accelerated mappings) Let $T:\real_{+}^N\to\real_{++}^N$ be a positive concave mapping and $\signal{M}$ be its lower bounding matrix. If $\rho(\signal{M})<1$, the accelerated mapping $\ta:\real_{+}^N\to\real_{++}^N$ of $T$ is the mapping given by:
\begin{align}
\label{eq.ta}
\ta(\signal{x}):=(\signal{I}-\signal{M})^{-1}(T(\signal{x})-\signal{Mx}).
\end{align}
\end{definition}

To see that the codomain of $\ta$ in the above definition is indeed $\real^N_{++}$, note that, by Lemma \ref{lemma.lb}, we have that $T(\signal{x})-\signal{M}\signal{x}\ge T(\signal{0})>\signal{0}$ for $\signal{x}\in\real_{+}^N$. Now use Fact~\ref{fact.neuman} to conclude that $(\signal{I}-\signal{M})^{-1}(T(\signal{x})-\signal{Mx})$ is a (strictly) positive vector for every $\signal{x}\in\real_{+}^N$.

The next lemma shows an alternative way to compute accelerated mappings. The main advantage of this alternative expression is computational. We have to perform only one matrix-vector multiplication.

\begin{lemma}
Let $\ta:\real_{+}^N\to\real_{++}^N$ be the accelerated mapping of the concave mapping $T:\real_{+}^N\to\real_{++}^N$, where we assume that the lower bounding matrix $\signal{M}$ of $T$ satisfies $\rho(\signal{M})<1$. Then $\ta$ in \refeq{eq.ta} can be equivalently expressed as 
\begin{align}
\label{eq.alternative}
\ta(\signal{x}) = (\signal{I}-\signal{M})^{-1}(T(\signal{x})-\signal{x})+\signal{x}.
\end{align}
\end{lemma}
\begin{proof}
Recalling that the matrix $\signal{I}-\signal{M}$ is invertible as a direct consequence of Fact~\ref{fact.neuman}, we deduce:
\begin{multline*}
\ta(\signal{x}) = (\signal{I}-\signal{M})^{-1}(T(\signal{x})-\signal{x})+\signal{x} \\ \Leftrightarrow
(\signal{I}-\signal{M})\ta(\signal{x})=T(\signal{x})-\signal{x}+(\signal{I}-\signal{M})\signal{x} \\ \Leftrightarrow (\signal{I}-\signal{M})\ta(\signal{x}) = T(\signal{x})-\signal{Mx} \\ \Leftrightarrow \ta(\signal{x})=(\signal{I}-\signal{M})^{-1}(T(\signal{x})-\signal{Mx}).
\end{multline*}
\end{proof}

Positive concave mappings and their corresponding accelerated mappings have many common characteristics. In particular, they are both standard interference mappings, and they have the same fixed point, as shown below.

\begin{lemma}
\label{lemma.shared}
Assume that $T:\real_{+}^N\to\real_{++}^N$ is a positive concave mapping with lower bounding matrix $\signal{M}$ satisfying $\rho(\signal{M})<1$. Then the accelerated mapping $\ta:\real_{+}^N\to\real_{++}^N$ of $T$ is a standard interference mapping. Furthermore, $\mathrm{Fix}(T)=\mathrm{Fix}(\ta)$.
\end{lemma}
\begin{proof}
By $\rho(\signal{M})<1$, the matrix inverse $(\signal{I}-\signal{M})^{-1}$ exists, and it is a non-negative matrix (Fact~\ref{fact.neuman}). As a result, each component of the mapping $T^\prime(\signal{x}):=(\signal{I}-\signal{M})^{-1} T(\signal{x})$ ($\signal{x}\in\real_{+}^N$) is a positive sum of concave functions, hence the resulting function is also concave. Observing that linear functions are both concave and convex, we verify that each component of $\ta(\signal{x})=T^\prime(\signal{x})-(\signal{I}-\signal{M})^{-1}\signal{x}+\signal{x}>\signal{0}$ is a positive sum of concave functions, hence $\ta$ is a positive concave mapping. Proposition~\ref{prop.pos_conc} now shows that $\ta$ is a standard interference mapping. Consequently, $\ta$ has a unique fixed point, if it exists (Fact~\ref{fact.uniqueness}). If $\signal{x}^\star\in\mathrm{Fix}(T)$, then $\ta(\signal{x}^\star)=(\signal{I}-\signal{M})^{-1} (T(\signal{x}^\star)-\signal{Mx}^\star) =(\signal{I}-\signal{M})^{-1} (\signal{x}^\star-\signal{Mx}^\star)=(\signal{I}-\signal{M})^{-1} (\signal{I}-\signal{M}) \signal{x}^\star = \signal{x}^\star$. The converse is also immediate. Note that $T(\signal{x})=(\signal{I}-\signal{M})\ta{(\signal{x})}+\signal{M}\signal{x}$, hence ${T}(\signal{x}^\star)=\signal{x}^\star$ if $\ta(\signal{x}^\star)=\signal{x}^\star$, and we conclude that $\mathrm{Fix}(T)=\mathrm{Fix}(\ta)$.
\end{proof}

The practical implication of Lemma~\ref{lemma.shared} is that, to compute the fixed point of a positive concave mapping $T$, we can instead compute the fixed point of its accelerated version $\ta$ by using the standard iteration $\signal{x}_{n+1}=\ta(\signal{x}_n)$ shown in Fact~\ref{fact.inter_map}. In many applications, having a monotone sequence $\{\signal{x}_{n}\}_{n\in\Natural}$ is desirable, and a sequence of this type can be constructed with the standard fixed point iteration by starting the iterations from $\signal{x}_1=\signal{0}$ (see Fact~\ref{subfact.monotone}). For example, in network planning and optimization tasks, the fixed points of the concave mappings are estimates of the power allocation or of the load at the base stations \cite{Majewski2010,siomina12,feh2013, renato14SPM,renato14SIP,ho2015}. Therefore, even if the mapping has a fixed point, the network design is invalid if the power or load of any base station exceeds its physical limit. If the iterative algorithm produces a monotonically increasing sequence, we obtain a certificate that the design is invalid as soon as any element of the vector sequence exceeds its limit. It is particularly in these cases that the standard iteration with the accelerated mapping $\ta$ converges faster than the standard iteration with the original mapping $T$, in the following sense:

\begin{definition}\label{definition.faster} ({\it}Faster convergence)
\label{def.no_slower} Let $\{\signal{x}_n\}_{n\in\Natural}\subset \real^N$ and $\{\signal{y}_n\}_{n\in\Natural}\subset \real^N$ be two sequences converging to the same vector $\signal{u}^\star\in\real^N$. We say that the sequence $\{\signal{x}_n\}_{n\in\Natural}\subset \real^N$ converges faster than $\{\signal{y}_n\}_{n\in\Natural}\subset \real^N$ if $\|\signal{x}_n-\signal{u}^\star\|_\infty \le \|\signal{y}_n-\signal{u}^\star\|_\infty$ for every $n\in\Natural$. 
\end{definition}

With the above definition, we can now formally state the improvement obtained by using $\ta$ instead of $T$ with the standard iteration in Fact~\ref{subfact.monotone}.

\begin{proposition}
\label{proposition.faster}
Assume that $T:\real_{+}^N\to\real_{++}^N$ is a positive concave mapping with lower bounding matrix $\signal{M}\in\real^{N\times N}$ satisfying $\rho(\signal{M})<1$. Let $\ta:\real_{+}^N\to\real_{++}^N$ be the accelerated mapping of $T$. Consider the following two sequences: $\{\signal{x}^\prime_{n+1}:=T(\signal{x}^\prime_{n})\}_{n\in\Natural}$ and $\{\signal{x}^\dprime_{n+1}:=\ta(\signal{x}^\dprime_{n})\}_{n\in\Natural}$. Assume that both sequences start from the same vector $\signal{u}\in\real_{+}^N$; i.e.,  $\signal{u}=\signal{x}^\prime_1=\signal{x}^\dprime_1$. If $\{\signal{x}^\prime_{n}\}_{n\in\Natural}$ is monotonically increasing (resp. monotonically decreasing) in each component, then the following holds:
\begin{subproposition}
 $\{\signal{x}^\dprime_{n}\}_{n\in\Natural}$ is monotonically increasing (resp. monotonically decreasing) in each component.
\end{subproposition}

\begin{subproposition}
$\signal{x}_n^{\dprime} \ge \signal{x}_n^{\prime}$ (resp. $\signal{x}_n^{\dprime} \le \signal{x}_n^{\prime}$ ) for every $n\in\Natural$.
\end{subproposition}

\begin{subproposition} If the mapping $T$ has a fixed point (which, in particular, it is automatically guaranteed if $\{\signal{x}_n^\prime\}_{n\in\Natural}$ is monotonically decreasing in each component), then $\{\signal{x}^\dprime_{n}\}_{n\in\Natural}$  converges faster than $\{\signal{x}^\prime_{n}\}_{n\in\Natural}$ to $\signal{x}^\star\in\mathrm{Fix}(T)$, in the sense of Definition~\ref{def.no_slower}.
\end{subproposition}

\end{proposition}
\begin{proof}
 We prove the proposition only for monotonically increasing sequences (in each component). The proof for monotonically decreasing sequences can be obtained in a similar fashion by reversing all inequalities. \par 
\begin{enumerate}
\item Recall that, by Lemma~\ref{lemma.shared}, $\ta$ is a standard interference mapping, so, in light of  Fact~\ref{subfact.monotone}, we only need to prove that $\ta(\signal{u})\ge\signal{u}$ if $T(\signal{u})\ge\signal{u}$.

 By assumption, $T(\signal{u})-\signal{u}\ge \signal{0}$ and $\rho(\signal{M})<1$. In particular, by using Fact~\ref{fact.neuman} and non-negativity of $\signal{M}$, the last inequality implies that $(\signal{I}-\signal{M})^{-1}$ is a non-negative matrix. Consequently, from \refeq{eq.alternative}, we deduce:
\begin{align*}
\ta(\signal{u})=(\signal{I}-\signal{M})^{-1}(T(\signal{u})-\signal{u})+\signal{u}\ge\signal{u}.
\end{align*}

\item We show the result by using induction. Assume that $\signal{x}_n^\dprime\ge\signal{x}_n^\prime$ for a given $n\in\Natural$. From the definition of the mapping $\ta$ in \refeq{eq.ta}, we deduce:
\begin{align}
\label{eq.tm}
\ta(\signal{x}_n^\dprime)=\signal{x}^\dprime_{n+1}=T(\signal{x}_n^\dprime)+\signal{M}(\signal{x}^\dprime_{n+1}-\signal{x}_n^\dprime).
\end{align}
We have already proved that $\{\signal{x}_n^\dprime\}_{n\in\Natural}$ is monotonically increasing with the assumptions of the proposition (hence $\signal{x}^\dprime_{n+1}-\signal{x}_n^\dprime \ge \signal{0}$), $\signal{M}$ is a non-negative matrix, and $T$ is a mapping satisfying the monotonicity property of standard interference functions. Using these observations in \refeq{eq.tm}, we verify that:
\begin{align*}
\signal{x}_{n+1}^\dprime=\ta(\signal{x}_n^\dprime) \ge T(\signal{x}_n^\dprime) \ge T(\signal{x}_n^\prime)=\signal{x}_{n+1}^\prime.
\end{align*}
The above arguments are valid, in particular, for $n=1$, because $\signal{x}_1^\dprime=\signal{x}_1^\prime=\signal{u}$ by assumption.

\item First recall that both $\{\signal{x}^\prime_{n}\}_{n\in\Natural}$ and $\{\signal{x}^\dprime_{n}\}_{n\in\Natural}$ converge to the uniquely existing fixed point $\signal{x}^\star\in\mathrm{Fix}(T)$ (Lemma~\ref{lemma.shared} and Fact~\ref{subfact.monotone}). The desired result $\|\signal{x}_n^\dprime-\signal{x}^\star\|_\infty \le \|\signal{x}_n^\prime-\signal{x}^\star\|_\infty$, valid for every $n\in\Natural$, follows directly from Proposition~\ref{proposition.faster}.2. 
\end{enumerate}
\end{proof}

As an immediate consequence of Proposition~\ref{proposition.faster} and Fact~\ref{subfact.monotone}, we have the following.

\begin{Cor}
Assume that $T:\real_{+}^N\to\real_{++}^N$ is a positive concave mapping with $\signal{x}^\star\in\mathrm{Fix}(T)\ne\emptyset$, and denote by $\ta:\real_{+}^N\to\real_{++}^N$ its corresponding accelerated mapping. Then the sequence $\{\ta^n(\signal{0})\}_{n\in\Natural}$ converges faster to $\signal{x}^\star$ than the sequence $\{T^n(\signal{0})\}_{n\in\Natural}$, in the sense of Definition~\ref{definition.faster} (we assume that both sequences start from the vector $\signal{0}$).
\end{Cor}

\begin{remark} Following Yates' arguments \cite{yates95} to prove the convergence of the iteration in Fact.~\ref{subfact.monotone}, we can also argue that the proposed accelerated scheme is expected to be fast when the initial point is arbitrary.  More precisely, assume that $T:\real^N_{+}\to\real_{++}^N$ is a positive concave mapping with a fixed point denoted by $\signal{x}^\star\in\real_{++}^N$. Since this fixed point is strictly positive, for an arbitrary vector $\signal{x}\in\real_{+}^N$ there always exists $\alpha > 1 $ satisfying $\signal{x} \le \alpha \signal{x}^\star$. From Definition~\ref{definition.inter_func}, we verify that $T(\signal{x}) \le T(\alpha \signal{x}^\star) < \alpha T(\signal{x}^\star)=\alpha \signal{x}^\star$. These inequalities imply that (see also Fact.~\ref{subfact.monotone}) i) $T^n(\signal{0}) \le T^n(\signal{x})\le T^n(\alpha \signal{x}^\star)$ for every $n\in\Natural$, ii) the sequence $\{T^n(\alpha \signal{x}^\star)\}_{n\in\Natural}$ is monotonically decreasing, and iii) $\{T^n(\signal{0})\}_{n\in\Natural}$ is monotonically increasing. In other words, each term of the monotone sequences $\{T^n(\signal{0})\}_{n\in\Natural}$ and $\{T^n(\alpha \signal{x}^\star)\}_{n\in\Natural}$ are, respectively, (element-wise) lower and upper bounds for each term of the sequence $\{T^n(\signal{x})\}_{n\in\Natural}$. All the above arguments are also valid if we exchange $T$ by its corresponding accelerated mapping $\ta$, and we note that lower and upper bounding sequences $\{\ta^n(\signal{0})\}_{n\in\Natural}$ and $\{\ta^n(\alpha \signal{x}^\star)\}_{n\in\Natural}$ for the sequence $\{\ta^n(\signal{x})\}_{n\in\Natural}$ converge faster to the fixed point $\signal{x}^\star$ when compared to the lower and upper bounding sequences $\{T^n(\signal{0})\}_{n\in\Natural}$ and $\{T^n(\alpha \signal{x}^\star)\}_{n\in\Natural}$ for the sequence $\{T^n(\signal{x})\}_{n\in\Natural}$. In other words, the sequence produced by $\signal{x}_{n+1}^\prime=\ta(\signal{x}_n^\prime)$ with $\signal{x}_1^\prime=\signal{u}\in\real_{+}^N$ arbitrary has sharper element-wise bounds than the sequence produced by $\signal{x}^\dprime_{n+1}=T(\signal{x}^\dprime_n)$ for the same starting point $\signal{x}_1^\dprime=\signal{u}$.

\end{remark}

\begin{remark} The price we pay to use the accelerated iteration ${\signal{x}}^\prime_{n+1}=\ta(\signal{x}^\prime_n)$ instead of using ${\signal{x}}_{n+1}=T(\signal{x}_n)$ is the need for a matrix-vector multiplication, if $\ta$ is evaluated by using \refeq{eq.alternative} (assuming that the lower bounding matrix is not the zero matrix). Furthermore, a matrix inversion is required (or, for increased numerical stability, a matrix decomposition), but this operation needs to be done only once. One situation where the proposed scheme is particularly useful is when the evaluation of the mapping $T$ is time consuming when compared to the matrix-vector multiplication. In this situation, for all practical purposes, the time to compute $\signal{x}^\prime_n$ or $\signal{x}_n$ is roughly equivalent for a given $n\in\Natural$ sufficiently small. However, for every $n\in\Natural$, $\signal{x}^\prime_n$ is guaranteed to be a better approximation of the fixed point of the mapping $T$ than $\signal{x}_n$. This situation is common in network planning.
\end{remark}

\section{Applications in Network Planning and Optimization}
\label{sect.applications}

We now apply the general results in the previous sections to two concrete estimation problems in LTE networks. First, we consider the load estimation task discussed in \cite{Majewski2010,siomina12,feh2013,renato14SPM}, among other studies. Briefly, the objective is to determine the bandwidth required to satisfy the data rate demand of all users in the network, by assuming that the transmit power of all base stations is given. 

The second application we consider is the reverse of load estimation. For a given load allocation at the base stations, the objective is to estimate the power allocation inducing that load. This reverse problem has been motivated by the study in \cite{ho2015}, which has proved that using all available bandwidth is advantageous from various perspectives, and, in particular, from the perspective of transmit energy savings and interference reduction. That study also proves that there exists a standard interference mapping having as its fixed point the solution of the power estimation problem, and the study in \cite{renato14SIP} has shown that the interference mapping can take the form of a positive concave mapping.

\subsection{Load estimation}
\label{sect.load_planning}
 We consider an LTE network with $M$ base stations and $N$ users represented by elements of the sets $\setm=\{1,\ldots,M\}$ and $\mathcal{N}=\{1,\ldots,N\}$, respectively. The set of users connected to base station $i\in\setm$ is denoted by $\mathcal{N}_i$, and the data rate requirement of user $j\in\setn$ is given by $d_j>0$. The propagation loss between user $j\in\setn$ and base station $i\in\setm$ is denoted by $g_{i,j}>0$. Each base station $i\in\setm$ has $K$ resource units that can be assigned to users, and the transmit power per resource unit for each base station $i\in\setm$ is $p_i>0$. The reliable downlink data rate for each resource unit connecting base station $i\in\setm$ to user $j\in\setn$ is approximated by the following well-established interference-coupling model \cite{Majewski2010,siomina12,renato14SPM,feh2013,MajewskiTurkeHuangEtAl2007Analytical}:
  
\begin{align*}
 \omega_{i,j}(\signal{\nu},\signal{p})=B\log_2\left(1+\dfrac{p_i g_{i,j}}{\sum_{k\in\setm\backslash\{i\}}\nu_k p_k g_{k,j}+\sigma^2}\right),
\end{align*}
where $\sigma^2$ is the noise power per resource unit, $\signal{p}=[p_1,\ldots,p_M]^t$ is the downlink power vector per resource unit, $\signal{\nu}=[\nu_1,\ldots,\nu_M]^t$ is the load vector, and $B$ is the bandwidth per resource unit. Here, the load $\nu_i$ is fraction of the number of resource units in the time-frequency grid that users in the set $\setn_i$ require from base station $i$. For fixed power allocation $\signal{p}\in\real_{++}^M$, the load is the solution to the following system of nonlinear equations \cite{Majewski2010,siomina12,renato14SPM,feh2013}:
\begin{align}
\label{eq.system}
\begin{matrix}
\nu_1 = f_1(\signal{\nu},\signal{p})\\
\vdots \\
\nu_M = f_M(\signal{\nu},\signal{p}),
\end{matrix}
\end{align}
where 
\begin{align*}
\begin{array}{rcl}
f_i:~\real_{+}^M\times\real_{++}^M &\to&\real_{++} \\ (\signal{\nu},\signal{p})&\mapsto&\sum_{j\in\mathcal{N}_i}\dfrac{d_j}{K\omega_{i,j}(\signal{\nu},\signal{p})}.
\end{array}
\end{align*}

Note that, for each fixed $\signal{p}\in\real_{++}^M$ and $i\in\setm$, the function $h_{\signal{p},i}:\real_{+}^M\to\real_{++}:\signal{\nu}\mapsto f_i(\signal{\nu},\signal{p})$ is concave, hence the solution of \refeq{eq.system} with fixed $\signal{p}$ can be obtained by computing the fixed point of the positive concave mapping given by \cite{feh2013,renato14SPM}
\begin{align} 
\label{eq.tp}
T_{\signal{p}}(\signal{\nu}):=[h_{\signal{p},1}(\signal{\nu}),\ldots,h_{\signal{p},M}(\signal{\nu})]^t.
\end{align}
 Therefore, all the theory developed in the previous sections applies to this problem, and, in particular, the novel acceleration schemes for the computation of fixed points. Before proceeding with numerical examples of the acceleration schemes, we revisit known results related to this problem, and we show how the application-agnostic approaches developed in  Sect.~\ref{sect.infimum} and in Sect.~\ref{sect.acceleration} can be used to reach these known results in a more convenient way.

In particular, the authors of \cite{siomina12} construct a matrix by computing the values that the partial derivatives of the functions $h_{\signal{p},1},\ldots,h_{\signal{p},M}$ attain when a given component of the argument $\signal{\nu}$ of these functions goes to infinity. It has been shown in \cite{ho2014data} that the system of nonlinear equations in \refeq{eq.system} has a solution if and only if the spectral radius of this matrix proposed in \cite{siomina12} is strictly less than one. Using the terminology and results in Sect.~\ref{sect.infimum} and in Sect.~\ref{sect.acceleration}, we note that the matrix suggested in \cite{siomina12} is a particular case of a lower bounding matrix in Definition~\ref{definition.lower_bounding} constructed with the technique in Proposition~\ref{proposition.inf_2}. The fact that the spectral radius of this lower bounding matrix gives sufficient and necessary conditions to characterize  the existence of a solution of the nonlinear system is a direct consequence of the application-agnostic results in Proposition~\ref{proposition.spec_radius} and Proposition~\ref{prop.sufficiency}.

To be more precise, we can use \refeq{eq.derivation_supgrad} to construct the lower bounding matrix $\signal{M}_\signal{p}$ of the mapping $T_\signal{p}$ as follows:
\begin{multline*}
\signal{M}_\signal{p}= \\ \left[\begin{matrix} 
\lim_{h\to\infty} g_1^1(\signal{\nu}^\prime+h\signal{e}_1) & \cdots & \lim_{h\to\infty} g_M^1(\signal{\nu}^\prime+h\signal{e}_M) \\
\vdots & \ddots & \vdots \\
\lim_{h\to\infty} g_1^M(\signal{\nu}^\prime+h\signal{e}_1) & \cdots & \lim_{h\to\infty} g_M^M(\signal{\nu}^\prime+h\signal{e}_M) \\
\end{matrix}\right],
\end{multline*}
where $\signal{\nu}^\prime\in\real_{++}^M$ is arbitrary and $g_{k}^i(\signal{\nu})$ is the $k$th component of a supergradient of the function $h_{\signal{p},i}$ at an arbitrary point $\signal{\nu}=[\nu_1,\ldots,\nu_M]^t$.  By noticing that the function $h_{\signal{p},i}$ is differentiable in the interior of its domain, $g_{k}^i(\signal{\nu})$ is simply the partial derivative $\dfrac{\partial}{\nu_k}  h_{\signal{p},i}(\signal{\nu})$ for every  $\signal{\nu}\in\real_{++}^M$. As a result, we can verify that the lower bounding matrix of the mapping $T_\signal{p}$ is given by $\signal{M}_\signal{p}=\mathrm{diag}(\signal{p})^{-1}\signal{M}^\prime {\mathrm{diag}(\signal{p})}$, where 
\begin{align}
\label{eq.low_matrix}
[\signal{M}^{\prime}]_{i,k} = \begin{cases}
0, &\mathrm{if}~~ i= k \\
\sum_{j \in \mathcal{N}_i} \dfrac{\mathrm{ln}(2)  d_j g_{k,j}}{KB g_{i,j}} & \mathrm{otherwise}.
\end{cases}
\end{align}

(We can also obtain \refeq{eq.low_matrix} by constructing the lower bounding matrix with the approach in \refeq{eq.derivation_rec}.)

By Remark~\ref{remark.eigenvalues}, $\rho(\signal{M}_{\signal{p}})<1$ is equivalent to  $\rho(\signal{M}^\prime)<1$, and we note that $\signal{M}^\prime$ does not depend on the power allocation $\signal{p}$, a fact originally stated in \cite{ho2014data}. Therefore, to verify whether the mapping $T_\signal{p}$ has a fixed point by using the results in Proposition~\ref{proposition.spec_radius} and Proposition~\ref{prop.sufficiency}, we can compute $\rho(\signal{M}^\prime)$ instead of $\rho(\signal{M}_{\signal{p}})$. In other words, knowledge of $\rho(\signal{M}^\prime)$ is sufficient to determine whether the system of nonlinear equations in \refeq{eq.system} has a solution, as already stated in \cite{ho2014data} for this particular application.

Having the lower bounding matrix in closed form, we can now proceed to the numerical evaluations of the novel acceleration schemes. In the simulations we show here, we compare the accuracy of the load estimates generated by the standard iteration $\signal{\nu}_{n+1}=T_{\signal{p}}(\signal{\nu}_n)$ with its accelerated version $\signal{\nu}^\prime_{n+1}={T_\signal{p}}_\mathrm{A}(\signal{\nu}^\prime_n)$. Table~\ref{table.simulation} lists the main parameters of the network. 

 \begin{table}[ht]
     \caption{Network parameters of the simulation}
     \label{table.simulation}
     \centering    
     \begin{tabular}{ l  l }
         \hline
         Parameter & Value \\ \hline
         \hline
         Carrier frequency & 900 MHz \\
         Number of resource units ($K$) & 25 \\
         Transmit power per resource unit (${p}_i$, $\forall i\in\setm$) & 1.6W\\         
         System bandwidth ($K\cdot B$) & 5 MHz \\
         Noise power spectral density & -145.1 dBm/Hz \\
         Propagation model & Okumura-Hata, urban \\
         Antenna height of base stations & 30m \\
         Antenna height of the users & 1.5m \\
         Number of users ($N$) & 200 \\ 
         Number of base stations ($M$)& 25 \\
         Data rate of each user ($d_j$, $\forall j\in\setn$) &  768 kbps \\
         Dimension of the field & $2500$m$\times 2500$m\\
         User distribution & Uniformly distributed at random \\
         Base station distribution & Uniformly distributed \\
         \hline
     \end{tabular}
 \end{table}

 The figure of merit used in the comparisons is the expected normalized mean error (NME), which we define by
 \begin{align} 
 \label{eq.nmse}
 {e}_{\mathrm{NME}}(\signal{\nu}):=E[\|\signal{\nu}-\signal{\nu}^\star\|/\|\signal{\nu}^\star\|],
 \end{align}
  where $\signal{\nu}^\star\in\mathrm{Fix}(T_{\signal{p}})$. We approximate the expectation operator by averaging the results of 100 runs of the simulation, and in each simulation the positions of the users (and hence the propagation loss) are the random variables. All iterations start from the zero vector, and networks where the corresponding concave mapping does not have a fixed point are discarded. Therefore, the expectation in \refeq{eq.nmse} is conditioned to the fact that spectral radius of the lower bounding matrix is strictly smaller than one. 

Fig.~\ref{fig.results_load_planning} shows results obtained by using the iterative scheme in Fact~\ref{subfact.monotone} with the original mapping $T_{\signal{p}}$ and with its proposed accelerated version ${T_\signal{p}}_\mathrm{A}$. We verify that the mapping ${T_\signal{p}}_\mathrm{A}$ requires fewer iterations than ${T_\signal{p}}$ to obtain a given numerical precision, which is an expected result by considering Proposition~\ref{proposition.faster}.

\begin{figure}
 	\centering
 		\includegraphics[width=\linewidth]{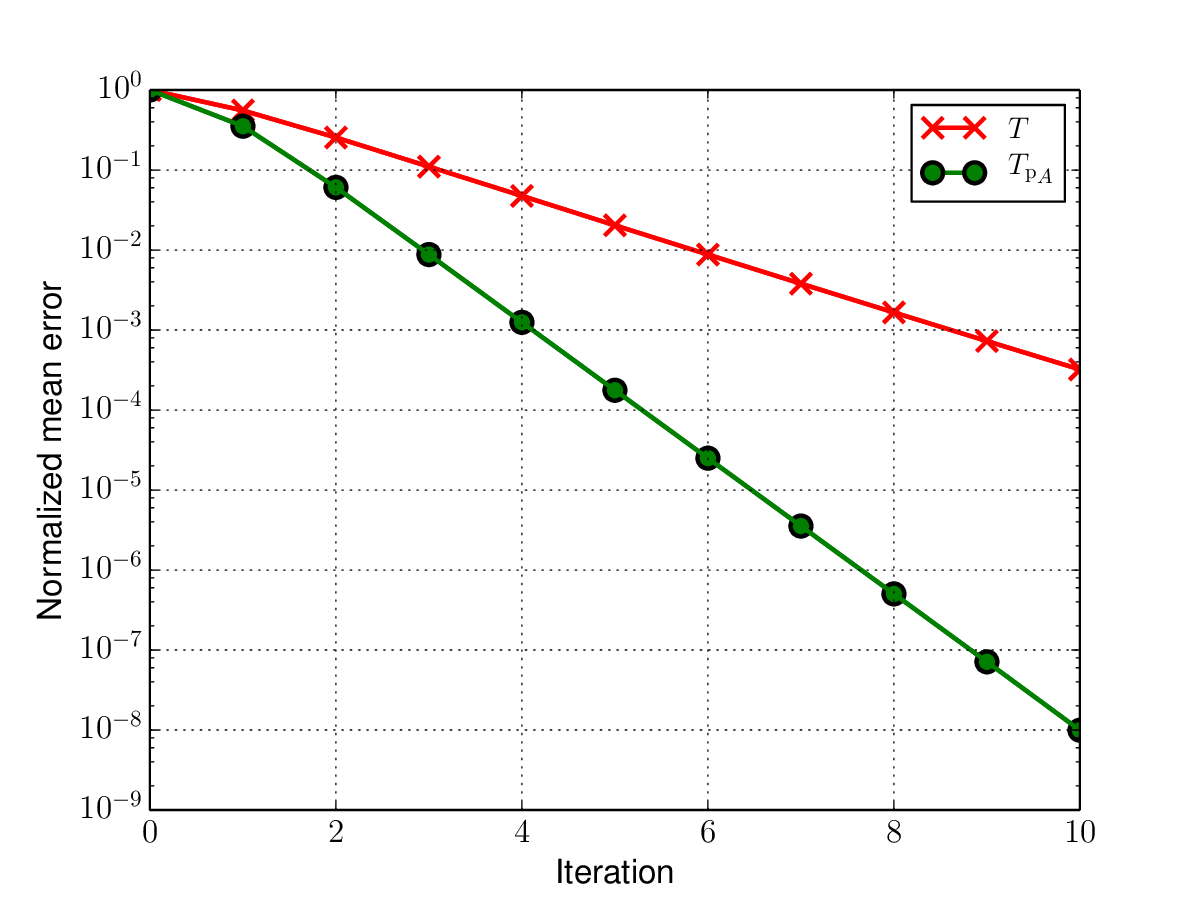}
	\caption{NME of the load estimate as a function of the number of iterations. Confidence intervals (95\%) have been computed, but they are not visible in the figure.} 
 	\label{fig.results_load_planning}
\end{figure}

\subsection{Power estimation}
\label{sect.power_planning}

We now turn our attention to the problem of power estimation in LTE networks. The objective is to solve \refeq{eq.system} for $\signal{p}=[p_1,\ldots,p_M]^t$ with the load $\signal{\nu}\in\real_{++}^M$ being the fixed parameter. It is shown in \cite{renato14SIP} that the solution of this nonlinear system is the fixed point of the positive concave mapping given by ${T}_{\signal{\nu}}(\signal{p}):=[{h}_{\signal{\nu},1}(\signal{p}),\ldots,h_{\signal{\nu},M}(\signal{p})]^t$, where 
 \begin{align*}
  {h}_{\signal{\nu},i}(\signal{p}):=\begin{cases}
       \dfrac{p_i}{\nu_i} \sum_{j\in\setn_i} \dfrac{d_j}{K\omega_{i,j}(\signal{\nu},\signal{p})},
 \quad \mathrm{if}~~ p_i\ne 0  \\ 
 \sum_{j\in\setn_i} \dfrac{d_j\ln 2}{KBg_{i,j}\nu_i}\left(\sum_{k\in\setm\backslash\{i\}}\nu_k p_k g_{k,j}+\sigma^2\right), \\ \qquad\qquad \mathrm{otherwise.}
      \end{cases}
 \end{align*}
 
 By using \refeq{eq.derivation_rec} to construct the lower bounding matrix $\signal{M}_{\signal{\nu}}$ of the mapping ${T}_{\signal{\nu}}$, we deduce:
 
 \begin{multline*}
 \signal{M}_\signal{\nu}= \\ \left[\begin{matrix} 
 \lim_{x\to 0^+} x {h}_{\signal{\nu},1}(x^{-1}\signal{e}_1) & \cdots & \lim_{x\to 0^+} x{h}_{\signal{\nu},1}(x^{-1}\signal{e}_M) \\
 \vdots & \ddots & \vdots \\
  \lim_{x\to 0^+} x{h}_{\signal{\nu},M}(x^{-1}\signal{e}_1) & \cdots & \lim_{x\to 0^+} x{h}_{\signal{\nu},M}(x^{-1}\signal{e}_M) 
 \end{matrix}\right] \\
 = \mathrm{diag}(\signal{\nu})^{-1} \signal{M}^\prime \mathrm{diag}(\signal{\nu}),
 \end{multline*}
 where $\signal{M}^\prime$ is the same matrix defined in \refeq{eq.low_matrix}. (The same result can be obtained by using Proposition~\ref{proposition.inf_2} to construct the lower bounding matrix, but here applying Proposition~\ref{proposition.inf_1} is easier than applying Proposition~\ref{proposition.inf_2}.)
 
 From Proposition~\ref{proposition.spec_radius} and the definition of $\signal{M}_{\signal{\nu}}$, we conclude that a necessary condition for existence of the fixed point of $T_{\signal{\nu}}$ is $\rho(\signal{M}^\prime)<1$,  which is the same requirement for the existence of the fixed point of $T_\signal{p}$ in \refeq{eq.tp}. However, there is a fundamental difference between these two mappings. As proved in \cite{ho2014data}, $\rho(\signal{M}^\prime)<1$ (note: this spectral radius does not depend on $\signal{\nu}$) is both a sufficient and necessary condition for the existence of the fixed point of $T_{\signal{p}}$. In contrast, the study in \cite{ho2015} has shown that the existence of the fixed point of $T_{\signal{\nu}}$ also depends on $\signal{\nu}$. 
 Therefore, ${T}_{\signal{\nu}}$ is an example of a mapping proving that the converse of Proposition~\ref{proposition.spec_radius} does not hold in general.
 
We now turn the attention to the acceleration schemes in this particular application. We use the same network considered in the load estimation task. The desired load $\signal{\nu}$ is obtained by solving \refeq{eq.system} with the power fixed to the value shown in Table~\ref{table.simulation}. Then we solve the reverse problem; we compute the power shown in Table~\ref{table.simulation} by using the standard iteration $\signal{p}_{n+1}=T_{\signal{\nu}}(\signal{p}_n)$ and its accelerated version $\signal{p}^\prime_{n+1}={T_\signal{\nu}}_\mathrm{A}(\signal{p}^\prime_n)$. Both algorithms start from the zero vector. The normalized mean error is again used as the figure of merit (which in this application is defined by replacing the load vector by the power vector in \refeq{eq.nmse}). We can see in Fig.~\ref{fig.results_power_planning} that in this application the proposed acceleration scheme once again provides us with clear advantages over the standard iterative approach, in accordance to the analysis in Sect.~\ref{sect.algorithm}. 

\begin{figure}
 	\centering
 		\includegraphics[width=\linewidth]{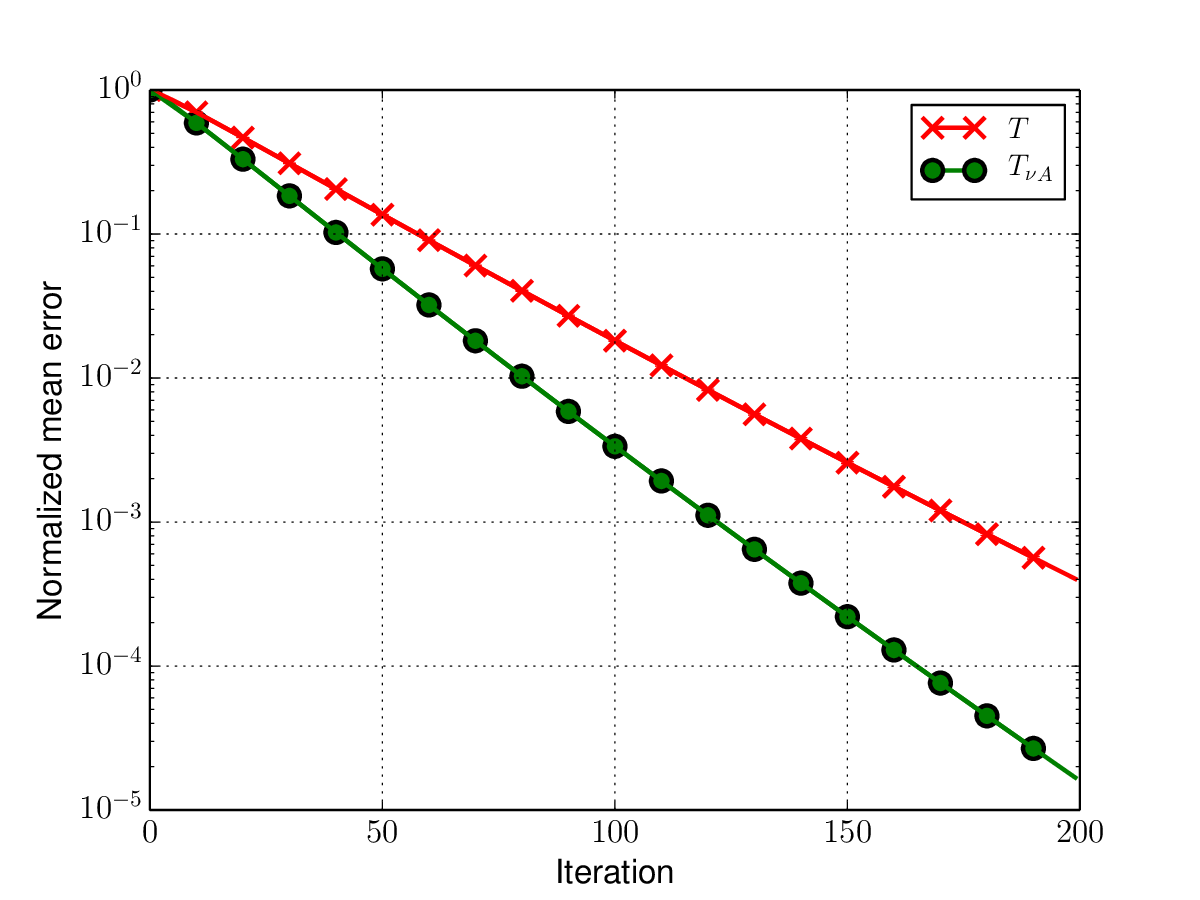}
	\caption{NME of the power estimate as a function of the number of iterations. Confidence intervals (95\%) have been computed, but they are not visible in the figure.} 
 	\label{fig.results_power_planning}
\end{figure}

\section{Conclusions}
We have shown that the results in \cite{siomina12} for the construction of lower bounding matrices in a very particular application domain can be generalized to a large class of positive concave mappings where even differentiability is not required. More specifically, we proved that positive concave mappings with nonempty fixed point set can be associated with a non-negative lower bounding matrix having spectral radius strictly smaller than one. By imposing additional assumptions on the mapping, having spectral radius strictly smaller than one also implies the existence of the fixed point of the concave mapping. We also demonstrated that the lower bounding matrix can be constructed with two simple and equivalent methods, and this matrix can be combined with its generating concave mapping to build a new mapping that preserves the fixed point. The standard fixed point iterations applied to this new mapping typically requires fewer evaluations of the original mapping to obtain an estimate of the fixed point for any given precision. The additional computational complexity of this novel approach is very modest. In the tasks of load and power estimation in LTE networks, where we are mostly interested in the precision of the estimates after a limited number of iterations, numerical examples show that the improvement in convergence speed obtained with the proposed method can be substantial. \par

\bibliographystyle{IEEEtran}
\bibliography{IEEEabrv,references}

\end{document}